\documentclass[runningheads]{llncs}

\usepackage[ruled,linesnumbered, noend, noline]{algorithm2e}
\usepackage{lineno,hyperref}
\usepackage[misc]{ifsym}
\usepackage{footmisc}
\usepackage{appendix}
\usepackage{amsmath}
\usepackage{times}
\usepackage{amsfonts}
\usepackage{fancyhdr}
\usepackage{graphicx}
\usepackage{xcolor}
\usepackage{lineno}
\usepackage{array}
\usepackage{longtable}
\usepackage[T1]{fontenc}
\usepackage{graphicx}
\def\x{{\textbf x}}

\def\o{{\textbf o}}
\def\s{{\textbf s}}
\def\u{{\textbf u}}
\def\SC{\mathsf{SC}}

\def\SC{\mathsf{SC}}

\def\kSC{\textsf{kSC}}
\def\FB{\textsf{FastSG}}
\def\FBo{\textsf{SGOpt}}

\def \y{{\textbf y}}
\def \0{{\textbf 0}}
\def \opt{{\mathsf{opt}}}

\begin{document}
	\title{Fast Stochastic Greedy Algorithm for $k$-Submodular Cover Problem}
	\titlerunning{Fast Stochastic Greedy Algorithm...}
	%
	\author{Hue T. Nguyen \inst{1,2} \and
		Tan D. Tran \inst{3} \and
		Nguyen Long Giang \inst{1}  \and Canh V. Pham (\Letter) \inst{3}}
	\authorrunning{Hue T. Nguyen  et al.}
	%
	\institute{Graduate University of Science and Technology
		\\
		Vietnam Academy of Science and Technology (VAST), Hanoi, Vietnam
		\email{huent@hau.edu.vn,nlgiang75@gmail.com}
		\\
		\and  Faculty of Information Technology, Hanoi Architecture University, Vietnam
		\\
		\and ORLab, Phenikaa University, Hanoi, 12116, Vietnam \\
		\email{tan.dinhtran@@phenikaa-uni.du.vn, canh.phamvan@phenikaa-uni.du.vn}
		\\
		Corresponding author: Canh V. Pham
	}

	\maketitle              
	\begin{abstract}
		We study the $k$-Submodular Cover ($\kSC$) problem, a natural generalization of the classical Submodular Cover problem that arises in artificial intelligence and combinatorial optimization tasks such as influence maximization, resource allocation, and sensor placement. Existing algorithms for $\kSC$ often provide weak approximation guarantees or incur prohibitively high query complexity. To overcome these limitations, we propose a \textit{Fast Stochastic Greedy} algorithm that achieves strong bicriteria approximation while substantially lowering query complexity compared to state-of-the-art methods. Our approach dramatically reduces the number of function evaluations, making it highly scalable and practical for large-scale real-world AI applications where efficiency is essential.
		\keywords{Combinatorial Optimization \and 	Approximation algorithm \and $k$-Submodular Cover.}
	\end{abstract}
	\section{Introduction}
	\label{sec:intro}
	The \textit{$k$-submodular optimization problems} have played an important role in both 
	\textit{operations research} and \textit{artificial intelligence}, with applications spanning a wide range of topics such as 
	influence maximization with multiple topics in social networks~\cite{ksub-nip15,ksub-icml20,ksub-stream-icml20,ksub-sdd,Canh_joco21}, 
	sensor placement~\cite{ksub-nip15,ksub-icml20,ksub-tevc18,ksub-sdd}, 
	feature selection~\cite{bi-sub-12}, 
	and information coverage maximization~\cite{ksub-tevc18,cor-24}, etc.
	Formally, let $V$ be a finite ground set and $k$ be a positive integer. We denote $[k]=\{1,2,\ldots,k\}$ and define  
	$
	(k+1)^V=\{(V_1,V_2,\ldots,V_k)\mid V_i \subseteq V, \ \forall i\in[k], \ V_i \cap V_j=\emptyset \ \forall i\neq j\},
	$
	which is the family of $k$ disjoint subsets of $V$, called a \textbf{$k$-set}.  
	A function $f:(k+1)^V \mapsto \mathbb{R}_+$ is called \textit{$k$-submodular} if, for any $k$-\textbf{set} $\x=(X_1,X_2,\ldots,X_k)$ and $\y=(Y_1,Y_2,\ldots,Y_k)$ in $(k+1)^V$, it holds that  
	$
	f(\x)+f(\y)  \geq f(\x \sqcap \y)+f(\x \sqcup \y),
	$
	where  
	\[
	\x \sqcap \y=(X_1 \cap Y_1,\ldots,X_k \cap Y_k),
	\]  
	and  
	\[
	\x \sqcup \y=(Z_1,\ldots,Z_k), \quad Z_i = X_i \cup Y_i \setminus \bigcup_{j \neq i}(X_j \cup Y_j).
	\]  
	In this paper, we study the \textbf{$k$-Submodular Cover ($\kSC$) problem}, defined as follows:
	\begin{definition}[$\kSC$ problem]
		Given a threshold $T\leq \max_{\x \in (k+1)^V}supp(\x)$, $\kSC$ asks to find $\x=(X_1, X_2, \ldots, X_k)\in (k+1)^V $ with \textbf{size of solution} $|\x|=\sum_{i \in [k]}|X_i|$ so that $f(\x)\geq T$. 
	\end{definition}
	When $k = 1$, the $\kSC$ reduces to the well-known Submodular Cover ($\SC$) problem~\cite{sc-nip23,snam,sc_in}. 
	This classical setting has been applied in diverse domains, such as threshold influence propagation in social networks~\cite{tap-17,vic-icml-2019}, 
	data summarization~\cite{nip-sc-aps,nips23-sc}, and revenue maximization in online social platforms~\cite{HartlineMS08,vic-icml-2019}.  
	For the general case with arbitrary $k$, $\kSC$ naturally serves as the dual formulation of the extensively studied 
	$k$-Submodular Maximization subject to a knapsack constraint~\cite{Canh_joco21,cor-24,TANG202228}. 
	It provides a unified abstraction for a wide range of practical tasks where the objective is to identify a minimum budget, size, cost 
	solution that guarantees the utility function surpasses a prescribed threshold.  Specifically, $\kSC$ arises in numerous significant applications such as 
	information diffusion and online advertising in social networks.
	
	A representative application is the \textbf{$k$-type Product Revenue Cover}. In this setting, an enterprise aims to promote products across multiple categories—such as electronics, fashion, or household items—while ensuring that the total revenue surpasses a predefined threshold $T$. Each customer contributes differently to the revenue depending on the product category, and each customer can be associated with at most one category. The goal is to select a subset of customers to target with promotional offers such that the aggregated revenue across all $k$ categories exceeds $T$, while minimizing the overall marketing cost. This problem naturally extends the classical submodular revenue maximization models into a multi-category, $k$-submodular framework.
	
	Another important instance is the \textbf{Influence Threshold with $k$-topics}. In viral marketing or product recommendation, an enterprise may wish to spread an advertising campaign simultaneously over $k$ different categories to influence at least $T$ individuals in a social network. To reduce expenses, they engage key opinion leaders or influential users to trigger this impact with minimal cost. This setting generalizes both the standard influence diffusion model for a single topic ($k=1$)~\cite{kem03,nguyen_sigmod16} and the influence threshold formulation~\cite{it-kuhnle}.
	
	Beyond these, $\kSC$ also naturally arises in other important domains, such as \textbf{multi-type sensor placement} for guaranteed information acquisition, and \textbf{multi-class feature selection} in machine learning. Together, these applications illustrate the broad relevance and impact of the $\kSC$ formulation across optimization, data mining, and artificial intelligence.
	
	Although the range of applications of $\kSC$ is extensive, 
	it is remarkable that only a handful of studies have explicitly examined this problem. 
	Furthermore, the established results for the classical $\SC$ problem are not directly transferable to $\kSC$ 
	because of the fundamental distinctions between $k$-submodularity and submodularity (see~\cite{icml-NguyenT20}). 
	This difference also accounts for why algorithms developed for $k$-submodular optimization 
	generally provide weaker approximation bounds compared to those in submodular optimization.  
	
	The earliest exploration of $\kSC$ was given in \cite{soict-kcover}, 
	where the presented algorithm achieves a bicriteria approximation that is comparatively weak in this scenario 
	(see Table~\ref{tab:1} for a summary).  
	Therefore, determining whether polynomial-time approximation algorithms for $\kSC$ 
	with tighter theoretical guarantees and practical efficiency exist 
	remains an open research challenge.  
	\begin{table}[hpt]
		\caption{Algorithms for $\kSC$ problem. The tuple $(x, y)$ presents $(x, y)$-bicriteria approximation ratio for an algorithm, i.e,  it returns a solution $\s$ such that $supp( \s) \leq y \cdot \opt$ and $f(\s) \geq x \cdot T$, 
			where $x, y > 0$.  Compared to the streaming algorithm in \cite{soict-kcover}, our proposed algorithm $\FB$ achieves a significantly stronger bicriteria approximation ratio. In particular, by setting $\delta=\epsilon^2$, it ensures a higher bound on solution quality, i.e., $|supp(\s)|\leq (1+\epsilon)\left(1+\ln\!\left(\tfrac{1}{\delta}\right)\right)$, while attaining the same utility value as the streaming approach, namely $f(\s)\geq \tfrac{1-\epsilon^2}{2}\opt$, where $\opt$ is the size of the optimal solution.}
		\centering
		\begin{tabular}{l p{5cm}p{3cm}}
			\hline
			\textbf{Reference} &   \textbf{Approximation ratio} & \textbf{Query Complexity}
			\\
			\hline
			$\FB$ (Alg.~\ref{alg:fb}, this work) & $\Big((1+\epsilon)(1+\ln(\frac{1}{\delta})), \frac{1-\delta}{2}\Big)$ & $O\left(\tfrac{nk\log^2 (n)}{\epsilon}  \log\!\left(\tfrac{n}{\delta}\right)\right)$	
			\\
			Streaming Algorithm in \cite{soict-kcover} & $\Big(\frac{1-\epsilon}{2\epsilon}, \frac{1-\epsilon^2}{2}\Big)$ & $O(nk\frac{\log(n)}{\epsilon})$
			\\
			\hline
		\end{tabular}
		\label{tab:1}
	\end{table}
	
	\textbf{Our contribution.} In this work, proposed  $\FB$ algorithm for $\kSC$ problem with the case $f$ is monotone. $\FB$ improves the existing bi-criteria approximation to $\Big((1+\epsilon)(1+\ln(\frac{1}{\delta})), \frac{1-\delta}{2}\Big)$ where $\epsilon, \delta \in (0,1)$ is constant parameters in $O(nk\log^3(n))$. In addition to our theoretical contributions, we conduct extensive experiments comparing our algorithm with the state-of-the-arts. Experimental results on real-world datasets further demonstrate that our approach consistently outperforms them in both running time and solution quality.
	Table~\ref{tab:1} compares our algorithms with the cutting-edge ones for $\kSC$ on two aspects including 
	bound of solution quality and query complexity.
	
	\textbf{Additional Related Works.}
	The study of $k$-submodular optimization was first initiated in the setting of bisubmodular maximization, corresponding to $k=2$, by \cite{bi-sub-12}. 
	Ward {\em et al.}~\cite{ksub_uncon_soda14} formally investigated the unconstrained $k$-submodular maximization problem. 
	In their seminal work, they proposed a deterministic greedy algorithm attaining an approximation ratio of $1/3$. 
	Later, \cite{ksub_uncon_soda16} improved this bound to $1/2$ for deterministic algorithms and further introduced a randomized greedy variant, 
	achieving an expected ratio of $k/(2k-1)$ by sampling elements according to a probability distribution. 
	Subsequently, \cite{ksub_iwoca-Oshima17} attempted to derandomize the algorithm of \cite{ksub_uncon_soda16} 
	to obtain the same ratio deterministically; however, their method required a large query complexity of $O(n^2k^2)$.  
	
	The $k$-submodular maximization problem has also been examined under cardinality and other combinatorial constraints. 
	\cite{ksub-nip15} studied monotone $k$-submodular maximization with two variants of cardinality restrictions: 
	(i) a global constraint limiting the overall solution size, and 
	(ii) individual constraints restricting the size of each partition $i \in [k]$. 
	They demonstrated that greedy approaches yield a ratio of $1/2$ in the global case and $1/3$ in the individual case. 
	Subsequent works further refined algorithms for cardinality-constrained settings, including 
	multi-objective evolutionary techniques~\cite{ksub-tevc18}, regret-bound analyses~\cite{ksub-Soma19}, 
	and streaming algorithms~\cite{ksub-stream-icml20,Ene22}.  
	
	In addition, \cite{ksub-jdo16} proposed a greedy approach for $k$-submodular maximization under matroid constraints, 
	achieving a ratio of $1/2$. 
	The knapsack constraint, which generalizes the cardinality constraint, has also been considered. 
	\cite{TANG202228} first derived a $(1/2 - 1/(2e))$-approximation algorithm for this case, inspired by the greedy framework of \cite{Sviridenko04}, 
	though it required $O(n^4k^3)$ oracle queries. 
	This ratio was later improved to $1/2$ by \cite{ksub-knap} through a multilinear extension method with polynomial query complexity. 
	Most recently, \cite{cor-24} presented a practical algorithm attaining a ratio of $1/3$ 
	with nearly linear query complexity $O(nk)$.

	The $\kSC$ problem was first investigated in~\cite{soict-kcover}, 
	where the authors proposed a $(\frac{1-\epsilon}{2\epsilon}, \frac{1-\epsilon^2}{2})$-bicriteria approximation. 
	Nevertheless, this result offers rather weak guarantees and the algorithms may not run in polynomial time, 
	often leading to arbitrarily large bounds.  
	In contrast, our algorithms remedy these drawbacks by providing stronger bicriteria ratios together with improved query complexity.  
	
	\textbf{Organization.}
	Section~\ref{sec:pre} introduces the notations and preliminary concepts. In Section~\ref{sec:alg-fbi}, we present our proposed algorithms along with their theoretical analysis. Section~\ref{sec:exp} reports the experimental results, demonstrating the effectiveness of our approach in practice. Finally, Section~\ref{sec:con} concludes the paper.
	\section{Preliminaries}
	\label{sec:pre}
	
	In this section, we introduce the notations used throughout the paper.  
	
	\textbf{$k$-set notions.} Given a ground set $V=\{e_1, e_2, \ldots, e_n \}$ and an integer $k$, we define $[k]=\{1, 2, \ldots, k\}$ and let $(k+1)^V=\{(V_1, V_2, \ldots, V_k)\mid V_i \subseteq V \ \forall i \in [k], \ V_i\cap V_j=\emptyset \ \forall i \neq j\}$ be the collection of $k$ disjoint subsets of $V$, called a $k$-\textbf{set}.  
	
	For $\x=(X_1, X_2, \ldots, X_k)\in (k+1)^V$, we define $supp_i(\x)=X_i$, $supp(\x)=\cup_{i\in [k]}X_i$, and call $X_i$ the \textbf{$i$-th set of $\x$}. An empty $k$-set is denoted by $\0=(\emptyset, \ldots, \emptyset)$.  
	If $e \in X_i$, then $\x(e)=i$ and $i$ is referred to as the \textbf{position} of $e$ in $\x$, otherwise $\x(e)=0$. Adding an element $e \notin supp(\x)$ into $X_i$ is expressed as $\x \sqcup (e, i)$.  
	We also write $\x=\{(e_1, i_1), (e_2, i_2), \ldots, (e_t, i_t)\}$ for $e_j \in supp(\x), \ i_j=\x(e_j), \ \forall 1 \leq j \leq t$.  
	When $X_i=\{e\}$ and $X_j=\emptyset, \ \forall j \neq i$, $\x$ is denoted by $(e,i)$.  
	
	For $\x=(X_1, X_2, \ldots, X_k), \y=(Y_1, Y_2, \ldots, Y_k)$ in $(k+1)^V$, we denote $\x \sqsubseteq \y$ iff $X_i \subseteq Y_i$ for all $i\in [k]$.  
	For simplicity, we assume that $f$ is non-negative, i.e., $f(\x)\geq 0$ for all $\x\in (k+1)^V$, and normalized, i.e., $f(\0)=0$.  
	
	\textbf{$k$-submodular function.} 
	A function $f: (k+1)^V \mapsto \mathbb{R}_+$ is called $k$-\textbf{submodular} if for any 
	$\x=(X_1, X_2, \ldots, X_k)$ and $\y=(Y_1, Y_2, \ldots, Y_k)$ in $(k+1)^V$, the following holds:
	\begin{align}
	f(\x)+ f(\y) \geq f(\x \sqcap \y) + f(\x \sqcup \y).
	\end{align}
	Here,
	\begin{align}
	\x \sqcap \y&=(X_1 \cap Y_1, \ldots, X_k \cap Y_k), \\
	\x \sqcup \y&=(Z_1, \ldots, Z_k), \ \mbox{where } Z_i = X_i \cup Y_i \setminus \bigcup_{j \neq i} (X_j \cup Y_j).
	\end{align}
	We assume the existence of an \textbf{oracle query}, which, when given a $k$-set $\x$, returns the value $f(\x)$. 
	We also recall some fundamental properties of $k$-submodular functions that will be useful in designing our algorithms.  
	The function $f$ is \textbf{monotone} if for any $\x \in (k+1)^V$, $e \notin supp(\x)$ and $i \in [k]$, 
	the \textit{marginal gain} of adding element $e$ to the $i$-th set $X_i$ of $\x$ is nonnegative, i.e.,
	\begin{align*}
	\Delta_{(e, i)} f(\x) &= f(X_1, \ldots, X_{i-1}, X_i \cup \{e\}, X_{i+1}, \ldots, X_k) 
	- f(X_1, \ldots, X_k) \geq 0.
	\end{align*}
	As shown in~\cite{ksub_uncon_soda14}, $k$-submodularity of $f$ implies both orthant submodularity and pairwise monotonicity. 
	The function $f$ is \textbf{orthant submodular} if
	\begin{align}
	\Delta_{(e, i)} f(\x) \geq \Delta_{(e, i)} f(\y)
	\end{align}
	for any $\x, \y \in (k+1)^V$, $e \notin supp(\y)$, $\x \sqsubseteq \y$, and $i \in [k]$.  
	The function $f$ satisfies \textbf{pairwise monotonicity} if for any $i, j \in [k], \ i \neq j$: 
	\begin{align}
	\Delta_{(e, i)} f(\x) + \Delta_{(e, j)} f(\x) \geq 0.
	\end{align}
	In this work, for the $\kSC$ problem we only consider $k \geq 2$, since when $k=1$, 
	a $k$-submodular function reduces to a standard submodular function.  
	We denote an instance of the $\kSC$ problem by $(V, f, T)$, with $\o$ as an optimal solution 
	and its corresponding optimal cost $\opt =|supp(\o)|$.  
	
	\textbf{Bicriteria approximation algorithm.}  
	An algorithm is called a $(\delta_1, \delta_2)$-\textbf{bicriteria approximation algorithm} for the $\kSC$ problem 
	if it returns a solution $\s$ such that $c(\s) \leq \delta_1 \cdot \opt$ and $f(\s) \geq \delta_2 \cdot T$, 
	where $\delta_1, \delta_2 > 0$.  
	
	\section{Fast Stochastic Approximation Algorithm}
	\label{sec:alg-fbi}
	This section introduces  our $\FB$ Algorithm for $\kSC$ problem with the monotone utility function $f$. We first describe \FBo, a stochastic greedy algorithm with a guessed $\opt$, i.e,  optimal size, serving as a simplified version of $\FB$. 
	We then present the full $\FB$ algorithm, which operates without this assumption.
	\subsection{$\FBo$ Algorithm}
	$\FBo$ takes as input an instance $(V, g, k, T)$, a guessed value of the optimal solution $v$, and two accuracy parameters $\epsilon \in (0,1/2)$ and $\delta > 0$.
	We re-define a truncated objective function $f:=(k+1)^V \mapsto \mathbb{R}_+$ by $f(\cdot) = \min\{f(\cdot), T/2\}$ ($f$ is still $k$-submodular and monotone).
	The main idea of the algorithm is that employs random sampling over the set $V \setminus supp(\s)$ 
	to reduce the number of queries, assuming access to a guessed value $v$ of the optimal size. Accordingly,
	the algorithm performs within $\lceil \tfrac{v}{2}\log(\frac{1}{\epsilon}) \rceil$ iterations. In each iteration $j$, it samples a random subset $R^j$ of size $\Upsilon$ uniformly from the set of remaining elements $V \setminus supp(\s)$. Then, it selects a pair $(e,i)$ that maximizes the marginal gain $\Delta_{(e,i)} f(\s)$ over all $e \in R^j$ and $i \in [k]$, and updates the current solution by $\s \leftarrow \s \sqcup (e,i)$.
	The algorithm stops after $\lceil \tfrac{v}{2}\log(\frac{1}{\epsilon}) \rceil$ iterations and returns the final solution $\s$.
	The details pseudocode of $\FBo$ are described in Algorithm~\ref{alg-kSC:fb-opt}.
	\begin{algorithm}[h]
		\SetNlSty{text}{}{:}
		\KwIn{An instance $(V, g,  k, T)$, $\epsilon \in (0, 1/2), \delta>0$, $v$} 
		\KwOut{A solution $\s$.}
		$\s\leftarrow \0$,
		re-define the function $f:(k+1)^V\mapsto \mathbb{R}_+$ as	$f(\cdot)\leftarrow \min\{f(\cdot), T/2\}$
		\\
		\For{$j=0$ to $\lceil \frac{v}{2}\log(\frac{1}{\delta})\rceil$}
		{	
			$R^j$ $\leftarrow$ a random subset of size $\Upsilon$ uniformly sampled from $V\setminus supp(\s)$
			\\
			$(e, i) \leftarrow \arg\max_{ e\in R,  i\in [k]}\Delta_{(e, i) }f(\s)$
			\\
			$\s \leftarrow \s \sqcup (e, i)$
			
		}
		\Return $\s$
		\caption{Stochastic Greedy with Guessed $\opt$ (\FBo)}
		\label{alg-kSC:fb-opt}
	\end{algorithm}
	
	We define following notations according to the operation of Algorithm~\ref{alg-kSC:fb-opt}.
	\begin{itemize}
		\item[$\bullet$] $(e_j, i_j)$ as the $j$-th element added of the main loop of Algorithm~\ref{alg-kSC:fb-opt}.
		\item[$\bullet$] $\s=\{(e_1, i_1), \ldots, (e_t, i_t)\}$: the $k$-set $\s$ after ending the main loop, $t=|supp(\s)|$.
		\item[$\bullet$]  $\s^j=\{(e_1, i_1),\ldots, (e_j, i_j)\}$: the $k$-set $\s$ after adding $j$ elements $1\leq j\leq t$, $\s^0=\x$, $\s^t=\s$.
		\item[$\bullet$] For any $k$-set $\x$ with $|supp(\x)|=v$, we define
		$\x^j=(\x \sqcup \s^j ) \sqcup \s^j$, $\x^{j-1/2}=(\x \sqcup \s^j ) \sqcup \s^{j-1}$. 
		\item[$\bullet$] Define $X^j=supp(\x^{j-1})\setminus supp(\s^{j-1})$
		\item[$\bullet$] $\u^t=\{(u_1, i_1), (u_2, i_2), \ldots, (u_r,i_r) \}$ is a set of elements that are in $\x^t$ but not in $\s^t$, $r=|supp(\u^t)|$.
		\item[$\bullet$] $\u^t_l=\s^t \sqcup \{(u_1, i_1), (u_2, i_2), \ldots, (u_l,i_l) \},  1 \leq l\leq r$ and $\u^t_0=\s^t$.
	\end{itemize}
	We first establish the probability bound of the random sampling process in Lemma~\ref{lem:lem1}.
	\begin{lemma}
		For each iteration $j$ of Algorithm~\ref{alg-kSC:fb-opt}, if we set $\Upsilon=\min\left\lbrace n, \frac{n-j+1}{v-j+1}\log(\frac{n}{\delta})\right\rbrace $ we have $
		\Pr[O^j \cap X^j = \emptyset] \leq \frac{\delta}{n}$.
		\label{lem:lem1}
	\end{lemma}
	\begin{proof}
		Since $R^j$ is chosen uniformly at random from 
		$V \setminus supp(s^{j-1})$, 
		the probability that a particular element in $X^j$ is included in $R^j$ is 
		$\tfrac{|X^j|}{|V \setminus supp(s^{j-1})|}$. 
		By independence of selections, we have
		\begin{align}
		\Pr\!\big[R^j \cap X^j = \emptyset\big]
		&= \left(1 - \frac{|X^j|}{|V \setminus supp(s^{j-1})|}\right)^{|R^j|} \\
		&\le \exp\!\left(-\frac{|X^j|\,|R^j|}{|V \setminus supp(s^{j-1})|}\right) \\
		&= \exp\!\left(-\log\!\left(\frac{n}{\delta}\right)\right) 
		= \frac{\delta}{n}.
		\end{align}
	\end{proof}
	\begin{lemma}
		After iteration $j$ in Algorithm~\ref{alg-kSC:fb-opt}, assume that $R^j\cap X^j\neq \emptyset$
		, we have $f(\x)-f(\x^j)\leq f(\s^j)$.
		\label{lem:lem2}
	\end{lemma} 
	\begin{proof}
		Since $R^j \cap X^j \neq \emptyset$, i.e, there exists an element 
		$e \in supp(\x^{j-1}) \setminus supp(\s^{j-1})$ is randomly selected into $R^j$. 
		Define $\x^{j-1/2}=(\x \sqcup \s^j ) \sqcup \s^{j-1}$, and $\s^{j-1/2}$ as follows: If $e_j \in X^j$, then $\s^{j-1/2}=\s^{j-1} \sqcup (e_j, \x(e_j)) $. If $e_j \notin X^j$, $\s^{j-1/2}=\s^{j-1}$. We have
		\begin{align}
		f(\x^{l-1}) - f(\x^{l})
		& \leq f(\x^{l-1}) - f(\x^{l-1/2}) \label{ine:1} \\
		& \leq f(\s^{l-1/2}) - f(\s^{l-1}) \label{ine:12} \\
		& \leq f(\s^{l}) - f(\s^{l-1}), \label{ine:13}
		\end{align}
		where \eqref{ine:1} follows from the monotonicity of $f$, 
		\eqref{ine:12} from the $k$-submodularity of $f$, and 
		\eqref{ine:13} from the selection rule of the algorithm.
		Summing over $l=1,\ldots,j$, we obtain
		\begin{align}
		f(\x) - f(\x^j)
		&= \sum_{l=1}^j \big(f(\x^{l-1}) - f(\x^l)\big) \\
		&\leq \sum_{l=1}^j \big(f(\s^l) - f(\s^{l-1})\big) \\
		&\leq f(\s^j).
		\end{align}
		which completes the proof.
	\end{proof}
	We now show the performance guarantees of Algorithm~\ref{alg-kSC:fb-opt} in Theorem~\ref{theo:fba-opt}.
	\begin{theorem} For $\epsilon,\delta \in (0, 1)$ and an integer $v\leq n$, the Algorithm~\ref{alg-kSC:fb-opt} runs in at most	$O\left( k n \log (v \log(\frac{1}{\epsilon}))\log(\frac{n}{\delta})\right)$ queries  and with probability $1-\delta$, returns a solution $\s$ such that $|supp(\s)|=
		(1+\log(\frac{1}{\delta}))v$ and $ f(\s)\geq (1-e^{-\frac{v}{\opt}\log(\frac{1}{\delta})})T/2$.
		\label{theo:fba-opt}
	\end{theorem}
	\begin{proof}Let $t$ denote the number of iterations of the main \textbf{for} loop.  
		The required number of queries is at most
		\begin{align}
		\sum_{j \in [t]} k \frac{n-j+1}{v-j+1} \log\!\left(\tfrac{n}{\delta}\right) 
		&= k \sum_{i \in [t]} \frac{n-v+i}{i} \log\!\left(\tfrac{n}{\delta}\right) \qquad \mbox{(Let $i = v-j+1$)} \\
		&= k (n-v) \sum_{i \in [t]} \left(1 + \frac{1}{i}\right) \log\!\left(\tfrac{n}{\delta}\right) \\
		&= O\!\left(k n \log(t)\, \log\!\left(\tfrac{n}{\delta}\right)\right) \\
		&= O\!\left(k n \log\!\Big(v \log\!\left(\tfrac{1}{\epsilon}\right)\Big) \log\!\left(\tfrac{n}{\delta}\right)\right).
		\end{align}
		By Lemma~\ref{lem:lem1} and the union bound, we obtain
		\begin{align}
		\Pr[X^j \cap R^j = \emptyset, \ \forall j ] \leq \tfrac{\delta}{n} \, t \leq \delta.
		\end{align}
		Hence,
		\begin{align}
		\Pr[X^j \cap R^j \neq \emptyset, \ \forall j ] \geq 1 - \delta.
		\end{align}
		By applying Lemma~\ref{lem:lem2} we have:
		\begin{align}
		T - f(\s^j) &\le f(\o) - f(\s^j) \\
		&= f(\o) - f(\o^j) + f(\o^j) - f(\s^j) \\
		&\leq f(\s^j) + \sum_{e \in X^j} \Delta_{(e, \o(e))} f(\s^j) \\
		&\leq f(\s^j) + \opt \big(f(\s^{j+1}) - f(\s^j)\big). \label{ine-1}
		\end{align}
		Applying this inequality over $t$ iterations (under the condition that the event $X^j \cap R^j \neq \emptyset$ holds for all $j$), we obtain
		\begin{align}
		T-2f(\s^{t})&\leq(1-\frac{2}{\opt})(T-2f(\s^{t-1}))
		\\
		&	\leq e^{-\frac{2}{\opt}}(T-2f(\s^{t-1})) \label{ineq1}
		\\
		&  \leq e^{-\frac{2}{\opt}t}T \ \ \mbox{(By repeatedly applying inequality~\eqref{ineq1} for 
			$t$ iterations)}
		\\
		& \leq e^{-\frac{2}{\opt}\frac{v}{2}\log(\frac{1}{\delta})}T
		\\
		& =e^{-\frac{v}{\opt}\log(\frac{1}{\delta})}  T 
		\end{align} 
		which implies that $f(\s)\geq (1-e^{-\frac{v}{\opt}\log(\frac{1}{\delta})})T/2$.
	\end{proof}
	\subsection{$\FB$ Algorithm}
	We now turn our attention to the our main  $\FB$ algorithm (Algorithm~\ref{alg:fb}) by removing the assumption of a known estimate $v$.
	\begin{algorithm}[hpt]
		\SetNlSty{text}{}{:}
		\KwIn{An instance $(V, g,  k, T)$, $f$ is monotone, $\epsilon \in (0, 1/2), \delta>0$.} 
		\KwOut{A solution $\s$.}
		$O=\{v\in \mathbb{N}: 1 \leq (1+\epsilon)^{i}\} \leq n\}$
		\\
		\ForEach{$v \in O$}
		{
			$\s_v \leftarrow \FBo(V, f, k, v, \delta) $
		}
		$\s \leftarrow \arg\max \{f(\s_v): |supp(\s_v)|\leq v\log(\frac{1}{\epsilon}), v\in O \}$ \label{best-can}
		\\
		\Return $\s$
		\caption{$\FB$ Algorithm}
		\label{alg:fb}
	\end{algorithm}
	$\FB$ first constructs the candidate set $O=\{(1+\delta)^i : 1 \leq (1+\delta)^i \leq n, i \in \mathbb{N}\}$ that contains geometrically increasing values of $v$. For each $v \in O$, it then invokes the subroutine $\FBo(V,f,k,v,\delta)$ to compute a candidate solution $\s_v$. After all candidates are generated, $\FB$ selects the solution with the minimum support size among those satisfying $f(\s_v) \geq (1-\delta)T/2$ (Line~\ref{best-can}). Finally, the algorithm outputs the chosen solution $\s$ with the largest value of $f$ with the size at most $v\log(\frac{1}{\epsilon})$.
	\begin{theorem}
		For inputs $\epsilon, \delta \in (0, 1)$, 
		with probability at least $1-\delta$, the	Algorithm~\ref{alg:fb} returns a 
		$\Big((1+\epsilon)\log(\tfrac{1}{\delta}),\,\tfrac{1-\delta}{2}\Big)$-bicriteria approximation
		and has
		$O(\left(\tfrac{nk\log^2 (n)}{\epsilon}  \log\!\left(\tfrac{n}{\delta}\right)\right)$  query complexity.
		\label{theo:fb}
	\end{theorem}
	\begin{proof}
		Algorithm~\ref{alg:fb} calls $\FBo$ $|O|$ times, so its query complexity is:
		\begin{align}
		&	\sum_{v \in O}O\left( k n \log (v \log(\frac{1}{\epsilon}))\log(\frac{n}{\delta})\right) 
		\\
		& \leq  O\left( k n \log_{1+\epsilon}(n) \log (n \log(\frac{1}{\epsilon}))\log(\frac{n}{\delta})\right)
		=
		O\left( \frac{ k n\log^2(n)}{\epsilon} \log(\frac{n}{\delta})\right).
		\end{align}
		Since $1\leq \opt \leq n$, there exist a value of $j \in \mathbb{N}$ so that $\opt \leq v=(1+\epsilon)^j\leq \opt(1+\epsilon)$. By the theoretical bounds of $\FBo$ (Theorem~\ref{theo:fba-opt}) we have
		\begin{align}
		|supp(\s_v)| = \frac{v}{2}\log(\frac{1}{\epsilon})\leq \frac{1+\epsilon}{2}\log(\frac{1}{\epsilon})\opt.
		\end{align}
		And with probability at least $1-\delta$, 
		\begin{align}
		f(\s)&\geq (1-e^{-\frac{v}{\opt}\log(\frac{1}{\delta})})\frac{T}{2}
		\\
		& \geq (1-e^{-\frac{v}{\opt}\log(\frac{1}{\delta})})\frac{T}{2}\ \ \  \mbox{(Since $v \ge \opt$)}
		\\
		& \geq  \frac{(1-\delta)T}{2}.
		\end{align}
		By the searching rule of the algorithm, it always  finds a solution that satisfies the desired theoretical bounds.
	\end{proof}
	\section{Experiments}
	\label{sec:exp}
	We conduct an experimental comparison of the our $\FB$ algorithm against the GREEDY and STREAMING, the sate-of-the-art ones in~\cite{soict-kcover}, using three evaluation criteria: utility value (value of $f$), number of queries, and solution size. The experiments are performed on the \textbf{$k$-type Product Revenue cover} application with two datasets: a randomly generated Erdős–Rényi (ER) graph $(n=2000, k=3)$ with edge probability \(p=0.01\), and the Email network\footnote{\url{https://snap.stanford.edu/data/email-Eu-core.html}} $(n=1005, k=5)$.
	Given a ground set of $V$ of $n$ user in a social network. Given a $k$-set $\s=(S_1, S_2, \ldots, S_k) \in (k+1)^V$, we consider 
	the utility function as an extended revenue value in \cite{kuhnle_quickksubmono} is defined as
	$
	f(\s) = \sum_{u \in V} \sum_{i=1}^k \left( \sum_{v \in S_i} w_{uv} \right)^{\alpha_{u,i}},
	$
	where \(\alpha_{u,i} \in (0,1)\) models the sensitivity of customer \(u\) to product \(i\). $f$ is monotone and  $k$-submodular over $V$.
	For all experiments, the parameters are set to \(\epsilon = 0.1\) and \(\delta = 0.1\).
	\paragraph{Experiment Results.} 
	\textbf{Size of solution (Fig.~\ref{fig:exp}a,d).} FastSG maintains compact solutions, with the solution size growing linearly with the threshold $T$. On ER, FastSG produces solutions of size 8--27, consistently smaller than GREEDY (18--62) and close to STREAMING (18--29). On Email, $\FB$ requires only 1--4 elements to reach feasibility, while GREEDY needs up to 10 and STREAMING as many as 30. This demonstrates that FastSG achieves more compact solutions while still ensuring validity. \textbf{Number of queries (Fig.~\ref{fig:exp}b,e).} FastSG demonstrates a significant advantage in reducing the number of function queries. On ER, the number of queries required by FastSG is approximately three to four times smaller than GREEDY and comparable to or lower than STREAMING. On Email, the advantage is even clearer: $\FB$ requires only $10^3\!-\!10^4$ queries, while GREEDY and STREAMING need $10^4\!-\!10^5$. This highlights the superior efficiency of $\FB$ in terms of computational cost. \textbf{Running time (Fig.~\ref{fig:exp}c,f).} $\FB$ clearly outperforms both GREEDY and STREAMING in terms of execution time. On ER, the running time of $\FB$ ranges from 118 to 1543 seconds, significantly lower than GREEDY (409--3305 seconds) and STREAMING (1365--1679 seconds). On Email, FastSG requires only 1--15 seconds, while GREEDY takes 8--87 seconds and STREAMING exceeds 370--600 seconds. Overall, $\FB$ achieves speedups of about 3–12$\times$ over GREEDY and up to 40$\times$ over STREAMING, highlighting its superior efficiency and suitability for time-sensitive applications.
	\begin{figure}[h!]
		\centering
		\includegraphics[width=0.32\linewidth]{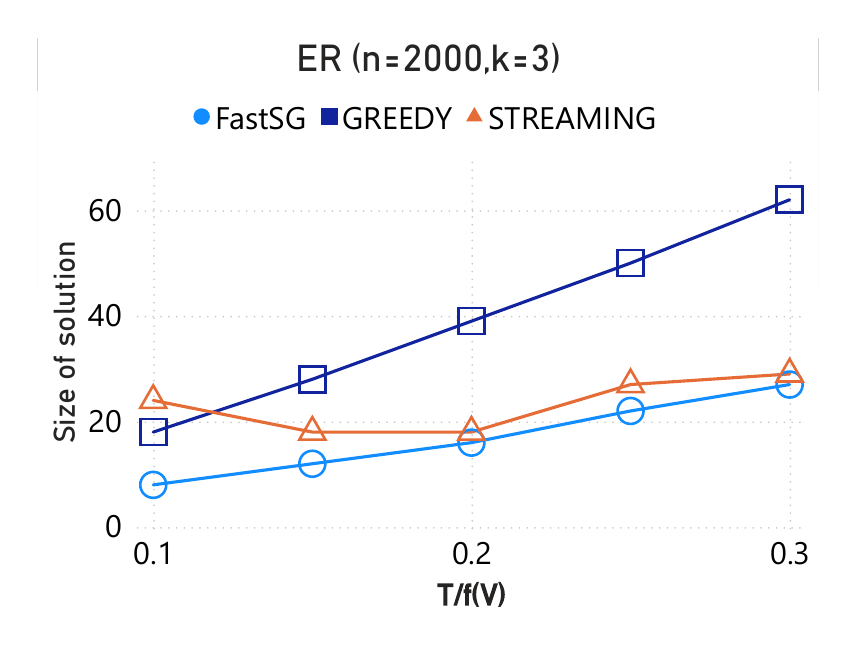}
		\includegraphics[width=0.32\linewidth]{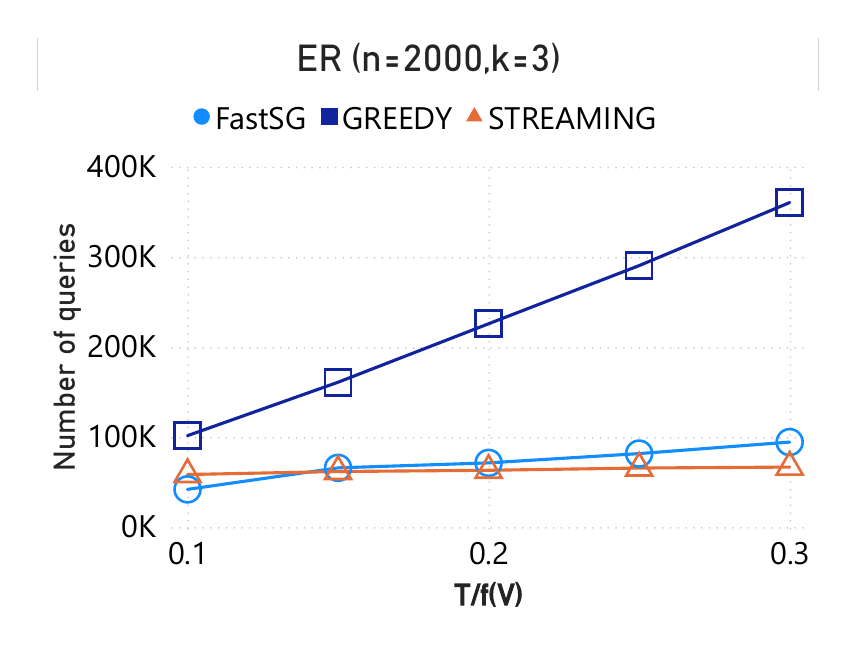}
		\includegraphics[width=0.32\linewidth]{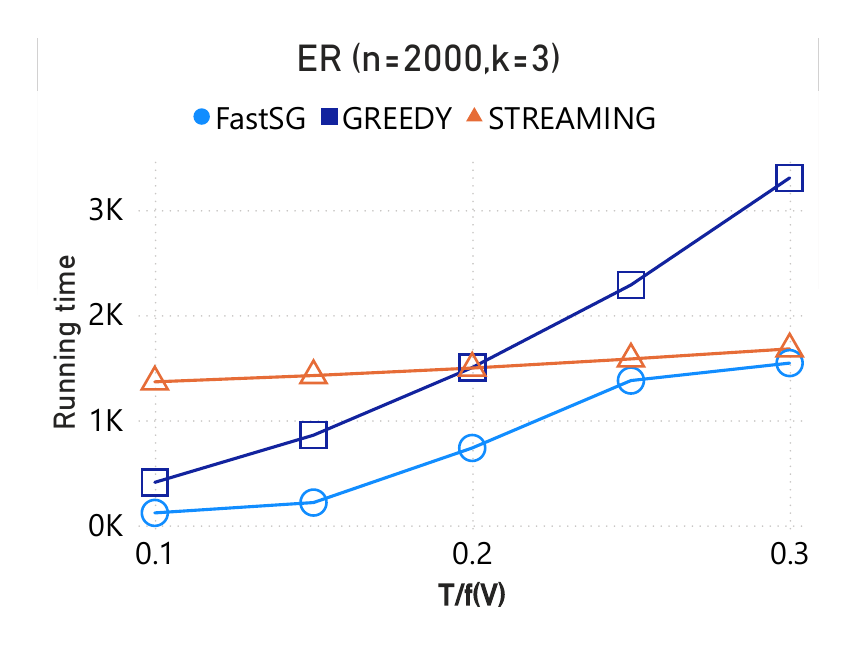}\\
		(a) \hspace{3.5cm} (b) \hspace{3.5cm} (c)\\
		\includegraphics[width=0.32\linewidth]{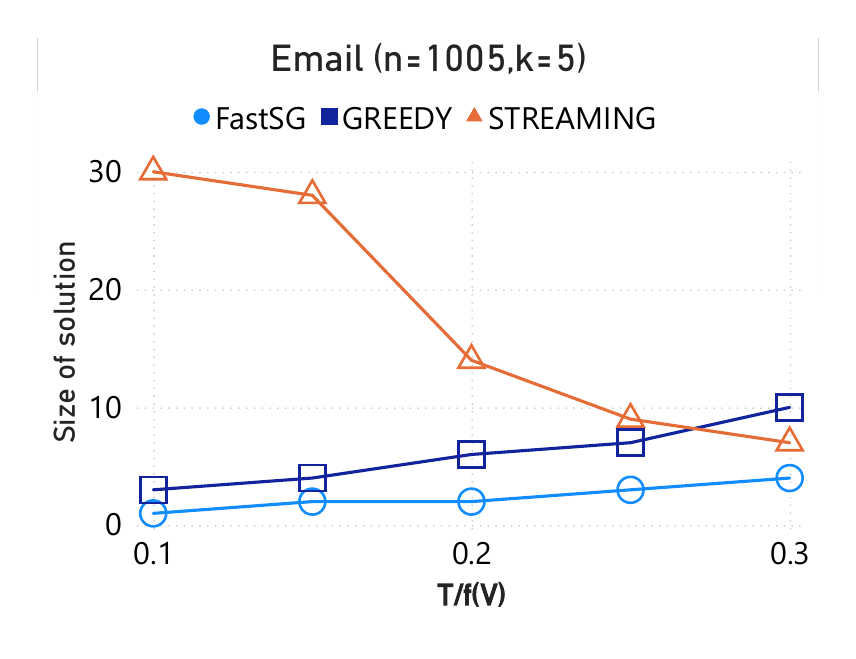}
		\includegraphics[width=0.32\linewidth]{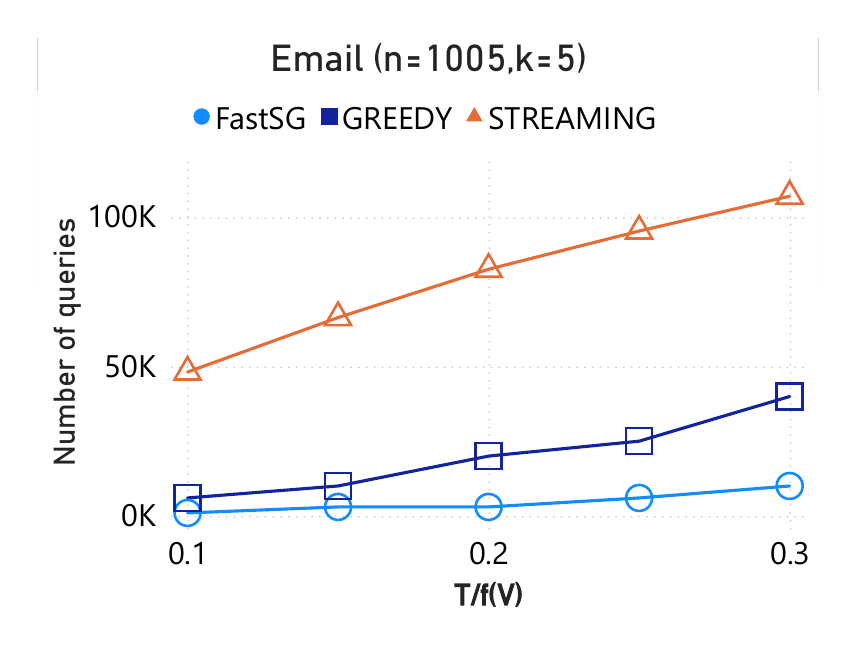}
		\includegraphics[width=0.32\linewidth]{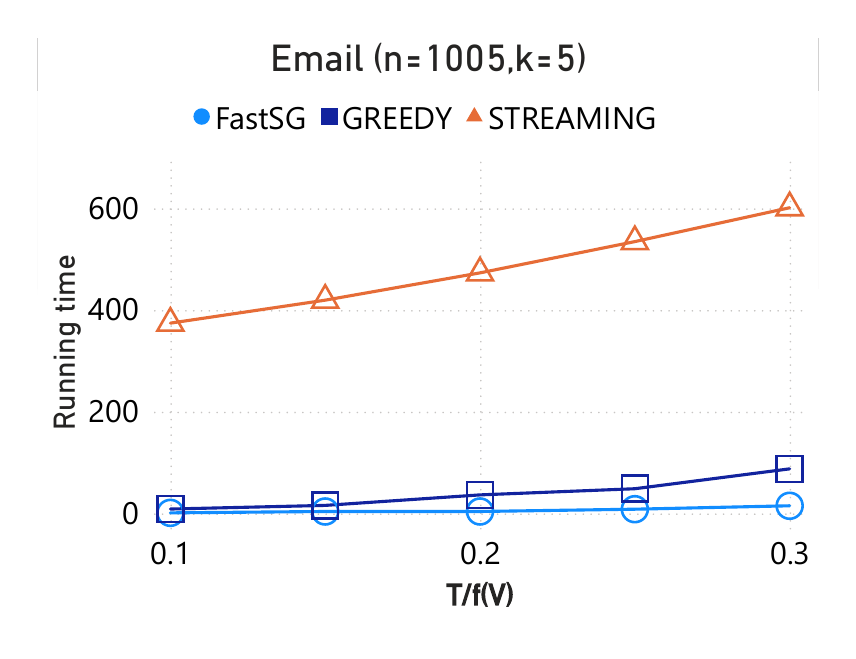}
		(d) \hspace{3.5cm} (e) \hspace{3.5cm} (f)\\
		\caption{Performance comparison (function value, query complexity, and size of solution) between FastSG and the state-of-the-art baselines GREEDY and STREAMING.}
		\label{fig:exp}
	\end{figure}
	\section{Conclusions}
	\label{sec:con}
	In this work, we have proposed the $\FB$ algorithm, which achieves the best-known bicriteria approximation guarantees for the $\kSC$ problem while requiring significantly lower query complexity. 
	Our $\FB$ algorithm substantially improves upon existing methods in both theoretical guarantees and practical performance. 
	An important open question remains: Can one design a bicriteria approximation algorithm that surpasses our result? 
	We leave this as a promising direction for future research.
	\bibliographystyle{splncs04}
	\bibliography{kSub-ref}
\end{document}